\documentclass[runningheads]{llncs}

\usepackage[utf8]{inputenc}

\usepackage{amsmath,amssymb,amsfonts}

\usepackage{graphicx}
\usepackage{latexsym}
\usepackage{algorithm}
\usepackage[noend]{algorithmic}

\usepackage{txfonts}
\usepackage{multirow,eepic}
\usepackage{wrapfig}
\usepackage{pdfpages}
\usepackage{refcount}
\usepackage{xspace}

\makeatletter
\def\senbun#1(#2)#3({\@senbun(#2)(}
\def\@senbun(#1,#2)(#3,#4){%
   \@tempdima#1\p@ \advance\@tempdima#3\p@
   \divide\@tempdima\tw@
   \@tempdimb#2\p@ \advance\@tempdimb#4\p@
   \divide\@tempdimb\tw@
   \edef\@senbun@temp{\noexpand\qbezier(#1,#2)%
      (\strip@pt\@tempdima,\strip@pt\@tempdimb)(#3,#4)}%
   \@senbun@temp}
\makeatother

\newcommand{\CB}{}

\newcommand{\ASY}{{\sc Asynch}}
\newcommand{\FSY}{{\sc Fsynch}}
\newcommand{\SSY}{{\sc Ssynch}}

\newcommand{\RR}{\textsc{Round-Robin}\xspace}
\newcommand{\CENT}{\textsc{Cent}\xspace}

\newcommand{\Look}{\mathit{Look}\xspace}
\newcommand{\Compute}{\mathit{Compute}\xspace}
\newcommand{\Move}{\mathit{Move}\xspace}
\newcommand{\LCM}{\mathit{LCM}\xspace}
\newcommand{\LU}{{\mathcal{LUMI}}} 
 
\newcommand{\FS}{{\mathcal{FST\!A}}} 
\newcommand{\FC}{{\mathcal{FCOM}}} 
\newcommand{\OB}{{\mathcal{OBLOT}}} 
\newcommand{\N}{{\rm I\kern-.22em N}} 
\newcommand{\Z}{{\sf Z\kern-.42em Z}} 
\newcommand{\R}{{\rm I\kern-.22em R}}

\newcounter{Codeline}

\newcommand\mycom[2]{\genfrac{}{}{0pt}{}{#1}{#2}}

\begin{document}

\title{Gathering Semi-Synchronously Scheduled Two-State Robots\thanks{This work was supported in part by JSPS KAKENHI Grant Number~20K11685, and~21K11748.}}
\titlerunning{Optimal Gathering on $\FS$ Mobile robots}

\author{Kohei Otaka\inst{1}\orcidID{0009-0008-2155-6960}
\and
Fabian Frei\inst{2,3}\orcidID{0000-0002-1368-3205}
\and\\
Koichi Wada\inst{1}\orcidID{0000-0002-5351-1459}
}
\authorrunning{K. Otaka et al.}

\institute{Hosei University, Tokyo, Japan\\
\email{kohei.otaka.8n@stu.hosei.ac.jp, wada@hosei.ac.jp} \and
ETH Zurich, Department of Computer 
Science, Zürich, Switzerland, \and CISPA Helmholtz Center for Information 
Security, Saarbrücken, 
Germany, \email{fabian.frei@inf.ethz.ch}, \email{fabian.frei@cispa.de}\\
}%

\maketitle

\begin{abstract}
%\color{blue} 
We study the problem \emph{Gathering} for $n$ autonomous mobile robots in synchronous settings
with a persistent memory called \emph{light}.
It is well known that Gathering is impossible in the basic model ($\OB$) where robots have no lights,
even if the system is semi-synchronous (called \SSY). 
Gathering becomes possible, however, if each robot has a light of some type that can be set to a constant number of colors.
In the $\FC$ model,
the robots can only see the lights of other robots.
In the $\FS$ model, 
each robot can only observe its own light. In the $\LU$ model, all robots can see all lights.
This paper focuses on $\FS$ robots with 2-colored lights in synchronous settings.
We show that 2-color $\FS$ and $\FC$ robots cannot solve Gathering in \SSY\ without additional conditions, even with rigid movement and agreement of chirality and the minimum moving distance. We also improve the condition of the previous gathering algorithm for $\FS$ robots with 2-color working in \SSY.
\end{abstract}

\section{Introduction}

\noindent\textbf{Background and Motivation.}\hspace*{1em}
The computational power of autonomously acting, simple, mobile robots has been the object of intense research in the field of distributed computing.
Ever since Suzuki and Yamashita's seminal work ~\cite{SY}, a large amount of work has been dedicated to the research of theoretical models of such autonomous mobile robots~\cite{AP,BDT,CFPS,DKLMPW,IBTW,KLOT,SDY}.
%FPS-book, etc
In the default setting, a robot is modeled as a point in the two-dimensional plane, and its capabilities are rather weak.
In particular, the robots are assumed to be \emph{oblivious} (have no memory to record past history), \emph{anonymous} (have no IDs), and \emph{uniform} (run identical algorithms) \cite{FPS}.  %in  synchronous   cycles. 

Robots operate in synchronous $\Look$-$\Compute$-$\Move$ ($\LCM$) cycles. 
In each \emph{round}, a nonempty set of (possibly all) robots is activated, 
and all activated robots then simultaneously perform an $\LCM$ cycle. The round ends as soon as all activated robots have performed their cycle. 
Each cycle is composed of three phases:
in the $\Look$ phase, a robot obtains a snapshot of the plane showing the positions of the other robots; in the $\Compute$ phase, it executes its algorithm (which is identical for all robots) using the snapshot as input; then it moves towards
the computed destination in the $\Move$ phase.
Repetition of these cycles allows robots to collectively perform some tasks and solve some problems. 

The selection of which robots are activated in a round is made by an adversarial scheduler.
This general setting is usually called {\em semi-synchronous}  (\SSY).
The special restricted setting where every robot is activated in every round is called  {\em fully-synchronous} 
 (\FSY)~\cite{FPS}. 

These systems have been extensively investigated within distributed computing.
The focus of the research has been on understanding the nature and extent of the impact of crucial factors, such as {\em memory persistence} 
 and {\em communication capability}, have on the solvability of a problem and thus on the computational power of the system.
 To this end, four robot models with light have been identified and investigated: $\OB$, $\FS$, $\FC$, and $\LU$.
% \color{black} Regarding the various kinds of light,
The most common (and weakest) model $\OB$~\cite{SY} (which stand for \emph{obl}ivious rob\emph{ot}s) assumes basic robots without light. 
In the strongest model $\LU$~\cite{DFPSY}(which stands for \emph{lumi}nous) robots can see their own lights as well as those of the other robots, 
whereas in $\FC$ and $\FS$~\cite{FSVY} (which stand for \emph{finite} \emph{com}munication and \emph{f}inite \emph{sta}tes, respectively), %\footnote{In \cite{FSVY}, external-light and internal-light are called FCOMM and FSTATE, respectively.}  
they can see, respectively, only the lights of the other robots (granting some communication capabilities) or only their own lights (which translates to persistent internal memory).

\noindent\textbf{Gathering and Previous Results.}\hspace*{1em}
\emph{Gathering} is one of the most fundamental tasks for autonomous mobile robots.
Gathering is the process where $n$ mobile robots, initially located in arbitrary positions, meet within finite time at an arbitrary single location.
When there are two robots (that is, for $n=2$), the task of Gathering is usually called \emph{Rendezvous}. 
Since %Gathering and 
Gathering is a simple but essential problem, it has been intensively studied  
and a number of possibility and impossibility results have been shown under different assumptions~\cite{AP,AOSY,BDT,CDN,CFPS,CP,DFPSY,DePP20,FPSW05,FSW19,IKIW,ISKIDWY,LMA,OWD,SDY,SY}.

%\input{introduction-table}
%% Table
%{\small
\begin{table}[t]
\centering
\caption{Previous Gathering algorithms for robots with lights.}
\label{tab:Table-Gathering}
{\footnotesize
\begin{tabular}{|c|c|c|c|c|c|}
\hline
Scheduler      & Movement  & $\LU$ & $\FC$ & $\FS$ & $\OB$ \\ \hline\hline
% FSYNC
\FSY & Non-Rigid & $\rightarrow$ & $\rightarrow$ & $\rightarrow$ & $\bigcirc$  \\ \hline
% SSYNC
\RR & Rigid & $\downarrow$ &  $\downarrow$ & 2~\cite{TWK} & $\times$~\cite{DePP20} \\ \hline
\CENT & Non-Rigid & $\downarrow$ & 2~\cite{TWK} & $\downarrow$ & $\times^*$\cite{DePP20} \\ \hline
\multirow{3}{*}{\SSY} & Rigid & $\downarrow$ & 3,2$^{**}$~\cite{TWK} & ? & $\times$\cite{FPS}\\ \cline{2-6}
               & Non-Rigid(+$\delta$=) & $\downarrow$ & ? & $2^{***}$~\cite{TWK} & $\uparrow$  \\ \cline{2-6}
               & Non-Rigid & 2~\cite{TWK} & ?& ? & $\uparrow$  \\ \hline
% ASYNC
\multirow{2}{*}{\ASY} & Rigid & $\downarrow$ & ? & $\downarrow$ & \multirow{2}{*}{$\uparrow$}  \\ \cline{2-5} 
               & Non-Rigid & 3\cite{NSW-2021} & ? & $\infty^{****}$~\cite{C04} &     \\ \hline
\end{tabular}

\smallskip

\noindent
\CB{
The symbols mean the following.\\ 
$^*$: Distinct gathering.  
$^{**}$: Local-awareness.
$^{***}$: $2\delta$-distant.
$^{****}$: unlimited number of colors.\\
$\bigcirc$: solvable.
$\times$: unsolvable.
?: unknown.
$\rightarrow, \downarrow, \uparrow$: The same as what the arrow is pointing to. (Here, the possibility and impossibility are derived from the stronger and weaker model, respectively.)
}
}
\end{table}

Table~\ref{tab:Table-Gathering} summarizes the previous results for Gathering by robots with lights. For all of them, it is assumed that different robots may have different local coordinate systems with different length units (i.e., no consistency between robots is guaranteed), but each local coordinate system remains the same throughout all rounds (i.e., the robots are all self-consistent); this is referred to as {\em fixed disorientation}. 
Moreover, no capabilities for detecting multiplicity are assumed.

In the basic $\OB$ model, Gathering is trivially solvable in \FSY\ but remains unsolvable in \SSY, even with assumptions such as rigidity of movement (i.e., robots always reach their target) or consistent chirality~\cite{FPS}. However, for robots equipped with lights, the problem becomes solvable in \SSY\ for various models like $\LU$, $\FC$, and $\FS$. Solvability in these settings results from different combinations of factors such as the number of available colors, algorithmic constraints, and movement restrictions. These differences highlight the complexity and versatility of solutions in the presence of different additional capabilities. Table~\ref{tab:Table-Gathering} also includes results for the asynchronous scheduler (\ASY) for comparison purposes. In \ASY, there is no common notion of time; the robots may be activated independently of the others, letting them perform their $\Look$, $\Compute$, and $\Move$ operations at arbitrary times~\cite{FPSW08}.

Regarding movement restriction, 
\emph{Rigid} means that robots always reach the computed destination during the movement operation.
\emph{Non-Rigid} means that a robot $x$ may be stopped before reaching the calculated destination but is be stopped before having moved some distance $\delta_x>0$ unknown to the robots, guaranteeing that any destination located within a radius of $\delta_x$ can be reached.
\emph{Non-Rigid(+$\delta$=)} is the same as  Non-Rigid, except that now $\delta_x$ is the same for all robots and known to them. 
In the following, we assume Non-Rigid(+$\delta=$).

We introduce some possibility and impossibility results for Gathering. 
It is known~\cite{FPS,DePP20} that Gathering for $\OB$ is not solvable in \SSY.
In particular, 
Gathering for $\OB$ is deterministically unsolvable in a restricted subclass of \SSY, where exactly one robot is activated in each round and they are always activated in the same order (called \RR).
Moreover, if all robots are initially located in different positions (called Distinct Gathering), it is deterministically unsolvable under
a $2$-bounded \CENT scheduler, where a scheduler is $2$-bounded if between any two consecutive activations of any robot, any other robot is activated at most $2$ times and \CENT means that exactly one robot is activated in each round~\cite{DePP20}. 
This impossibility holds even if we assume chirality and rigidity. %(Rigid, Non-Rigid(+$\delta$=), and Non-Rigid). 

Multiplicity detection is a strong assumption when it comes to solving Gathering. We know~\cite{FPS} that, if strong multiplicity is assumed, Gathering for $n$ robots is solvable in \SSY\ if and only if $n$ is odd, and that distinct Gathering for $n$ robots is solvable for $n\geq 3$ even in \ASY.
Non-oblivious robots have persistent memory. This is also true for robots with internal lights, but we restrict the amount of memory with internal lights is restricted to a constant. Gathering was already shown to be solvable with non-oblivious robots~\cite{C04}. %M2
However, the known algorithm stores the locations of other robots exactly, and the amount of memory exceeds any constant. It has remained unknown whether Gathering is solvable by robots with internal lights with a constant number of colors.

In the $\LU$ model, there is a 2-color algorithm with chirality and non-rigid movement in \SSY~\cite{TWK} and 3-color one with the same assumption in \ASY~\cite{NSW-2021}.
In the $\FC$ model, there is an algorithm with $3$ colors assuming rigidity. 
Assuming local awareness (i.e., robots recognize other robots sharing the same location) reduces the number of used colors down to $2$~\cite{TWK}.
In the $\FS$ model, there is an algorithm assuming Non-Rigid(+$\delta$=) with lights of only $2$ colors if %non-rigidity where the robots know the minimum moving distance $\delta$ and where 
the initial configuration of robots is $2\delta$-distant, where $2\delta$-distant means that the largest distance between two robots in the configuration is at least $2\delta$.
There also exist 2-color Gathering algorithms in \CENT and \RR schedulers 
for $\FC$ and $\FS$ robots, respectively~\cite{TWK}.

\noindent\textbf{Our Contributions.}\hspace*{1em} 
We prove the impossibility of Gathering on $\FC$ or $\FS$ for robots with 2 colors.
Specifically, we show that $\FC$ and $\FS$ robots with two-colored lights cannot solve Gathering in \SSY, even under the assumptions of rigid movement, consistent chirality, and a shared unit of length.
%agreement on the known minimum movement distance $\delta$. 
%TODO: rigid, but still minimum movement distance? Separate assumptions?
This result demonstrates that the assumptions of previous Gathering algorithms for $\FC$ and $\FS$ robots shown in Table~\ref{tab:Table-Gathering} are optimal in the sense that Gathering becomes impossible without assuming them.
For example, for the 2-color algorithm for $\FC$ robots working with rigid movement, the condition of {\em local-awareness} %\footnote{If any robot can observe other robots located on the same position, it is said to be local-aware.}
cannot be removed, and for the 2-color algorithm for $\FS$ robots working in Non-Rigid(+$\delta$=), the condition of {\em $2\delta$-distant} cannot be removed either. %, where $2\delta$-distant means the diameter of the configuration is at least $2\delta$.
The conditions that cannot be removed are minimal and should be as weak as possible. In this paper, we demonstrate that a Gathering algorithm for $\FS$ robots with two colors exists under conditions weaker than a $2\delta$-distant initial configuration. Specifically, we show that if only two forbidden patterns are excluded in the initial configuration, a Gathering algorithm for $\FS$ robots with two colors can be achieved.

%\medskip
%\subsection{Roadmap}
\noindent\textbf{Roadmap.}\hspace*{1em} 
The remainder of the paper is organized as follows. In Section~\ref{sec:model}, 
we define our robot model, the gathering problem, and the terminologies. 
Section~\ref{sec:Impossibility} presents the impossibility of Gathering for $\FC$ and $\FS$ robots with 2-color lights. 
Section~\ref{sec:algorithm} presents an optimal Gathering algorithm for $\FS$ robots with 2-color lights in \SSY.  Section~\ref{sec:conclusion} concludes the paper with a short summary.

\section{Preliminaries}\label{sec:model}
We consider a set of anonymous mobile robots $\mathcal{R}=\{r_{1},...,r_{n}\}$ located in $\mathbb{R}^{2}$. Each robot $r_{i}$ has a persistent state $l_{i}$ called its light, which may be taken from a finite set $L$ of colors. We denote by $l_{i}(t)\in L$ the color that the light of robot $r_i$ has at time $t$ and by $p_{i}(t)\in\mathbb{R}^{2}$ the position occupied by $r_i$ at time $t$ represented in some global coordinate system. A configuration $C(t)$ at time $t$ is a multiset of $n$ pairs $(l_{i}(t),p_{i}(t))$, each defining the color of light and the position of the robot $r_i$ at time $t$. When no confusion arises, $C(t)$ is simply denoted by $C$.

For a subset $S$ of $L \times \mathbb{R}^{2}$, $\mathcal{L}(S)$ and $\mathcal{P}(S)$ denote the projections to $L$ and $\mathbb{R}^{2}$ from $S$, respectively.

Each robot $r_i$ has its own coordinate system where $r_i$ is located at its origin at any time. These coordinate systems do not necessarily agree with those of other robots. This means that there is no guarantee of a common unit distance, nor for the directions of coordinate axes, nor for a clockwise orientation (chirality). However, each local coordinate system remains the same throughout all rounds. This is called the {\em fixed disorientation}. %We also assume there might not be consistency between the local coordinate systems and their unit of distance, but each local coordinate system remains the same throughout all rounds, called the {\em fixed disorientation}. 

At any time, any robot can be active or inactive. When a robot $r_{i}$ is activated, it executes the $\Look$, $\Compute$, and $\Move$ cycles:

\begin{itemize}
    \item \textbf{$\Look$:} The robot $r_i$ activates its sensors to obtain a snapshot which consists of a pair of light and position for every robot with respect to the coordinate system of $r_i$. 
    %TODO: What about multiplicity detection?
    Let $\mathcal{SS}_{i}(t)$ denote the snapshot of $r_i$ at time $t$. We assume that robots can observe all other robots (unlimited visibility). Note that $\mathcal{SS}_{i}(t)$ represents a sub-multi-set of $C(t)$ according to imposed assumptions in the local coordinate system of $r_i$, where $r_i$ is at the origin.
    \item \textbf{$\Compute$:} The robot $r_i$ executes its algorithm using the snapshot and (if visible) the color of its own light and returns a destination point $des_i$ expressed in its own coordinate system and a light $l_{i} \in L$. The robot $r_i$ sets its own light to the color $l_{i}$.
    \item \textbf{$\Move$:} The robot $r_i$ moves to the computed destination $des_i$. If the robot may be stopped by an adversary before reaching the computed destination, the movement is said to be {\em non-rigid}. Otherwise, it is said to be {\em rigid}. If stopped before reaching its destination, we assume that a robot has moved at least a minimum distance $\delta>0$. Note that without this assumption an adversary could make it impossible for any robot to ever reach its destination. If the distance to the destination is at most $\delta$, the robot can thus reach it. If the movement is non-rigid and robots know the value of $\delta$, this is called Non-Rigid($+\delta=$).
\end{itemize}

In the $\Look$ operation, the snapshot $\mathcal{SS}_{i}$ of $r_i$ should contain the positions of all robots, including $r_i$. However, if robots located on $p_i$ and $r_i$ can recognize the other robots, the robots have multiplicity detection at this point. Thus, we separately classify the observation of other robots located on $p_i$ for robot $r_i$. If any robot $r_i$ can observe the other robots located on $p_i$, it is said to be \emph{local-aware}. Otherwise, it is said to be local-unaware. Note that if we assume local awareness, $r_i$ recognizes whether other robots occupy location $p_i$ or not. In the following, we usually use the assumption that the system is local-aware.

A scheduler decides which subset of robots is activated for every configuration. The scheduler we consider is semi-synchronous. Moreover, it is always assumed that schedulers are fair, that is, each robot is activated infinitely often.

\begin{itemize}
    \item \textbf{\SSY:} The semi-synchronous scheduler (\SSY) activates a subset of all robots synchronously and their $\Look$-$\Compute$-$\Move$ cycles are performed at the same time. We can assume that activated robots at the same time obtain the same snapshot (adjusted to their local coordinate system) and their $\Compute$ and $\Move$ are executed instantaneously. In \SSY, we can assume that any activation happens in a discrete-time round and the $\Look$-$\Compute$-$\Move$ cycle is performed instantaneously in each round. In the following, since we consider \SSY\ and its subsets, we use round and time interchangeably.
\end{itemize}

As a special case of \SSY, if all robots are activated in each round, the scheduler is called fully-synchronous (\FSY). 

Let $C(t)$ be a configuration in round $t$. When $C(t)$ reaches $C(t+1)$ by executing the cycle at $t$, this is denoted as $C(t)\rightarrow C(t+1)$, where $C(t+1)$ is obtained by activating the robots once at time $t$ to execute the algorithm on $C(t)$. The reflective and transitive closure of $\rightarrow$ is denoted as $\rightarrow^{\ast}$. That is, a configuration transition $C(t) \rightarrow C(t^{\prime}) \rightarrow \cdots \rightarrow C(t^{\prime\prime})$ is denoted by $C(t) \rightarrow^{\ast} C(t^{\prime\prime})$.

Snapshots may be different by using assumptions even if these configurations are the same, and they depend on the multiplicity detection and on how robots can see lights of other robots when robots are equipped with lights. Robots are said to be capable of (weak) multiplicity detection if they can distinguish whether a point is occupied by at least two robots. The multiplicity detection is strong if the robots can detect the exact number of robots at any given point.

In our settings, robots have persistent lights and can change their color after $\Compute$ operation. With regard to the visibility of the lights, we consider the following robot model.

\begin{itemize}
    \item \textbf{$\LU$:} The robot can recognize not only colors of lights of other robots but also its own color of light.
    \item \textbf{$\FC$:} The robot can recognize only colors of lights of other robots but cannot see its own color of light. Note that a robot can still set its own color in each round.
    \item \textbf{$\FS$:} The robot can recognize only the color of its own light but not the lights of other robots.
\end{itemize}

When a robot performs the $\Look$ operation in $\FS$, its snapshot is the same as in the case of robots without lights.

Given a snapshot $\mathcal{SS}_{i}$ of a robot $r_i$ and a point $p_{j}(j\neq i)$ included in $\mathcal{P}(\mathcal{SS}_{i})$, a view $V_{i}[p_j]$ of $p_j$ in $\mathcal{SS}_{i}$ is a subset of $AL_{i}[p_j]=\{l | (l,p_j) \in \mathcal{SS}_{i},r_j \neq r_i\}$, where $AL_{i}[p_j]$ is a multi-set of colors of other robots that $r_i$ can see at point $p_j$. For any robot $r_i$ and any point $p$ in the snapshot of $r_i$, if $V_{i}[p]=AL_{i}[p]$, the view of the robots is called the \emph{multiset view}. If $V_{i}[p]$ regards $AL_{i}[p]$ as just a set, it is called \emph{set view}. If $V_{i}[p]$ is a set of any single element taken from $AL_{i}[p]$, it is called \emph{arbitrary view}. Let $V_i$ denote $\bigcup_{(l,p)\in \mathcal{SS}_{i}}V_{i}[p]$.

Multiset view is a strong assumption because robots without lights (i.e., with one color) can have strong multiplicity detection if multiset view is assumed. In fact, we can solve the Gathering problem by using robots without lights and multiset view \cite{flocchini2022distributed}. On the other hand, set view and arbitrary view do not imply multiplicity detection. In the following, we assume set view.

The $n$-Gathering task is defined as follows: given $n(\geq 2)$ robots initially placed at arbitrary positions in $\mathbb{R}^{2}$, let them congregate in finite time at a single location which is not predefined. In the following, the case $2$-Gathering problem is called $Rendezvous$ and the $n$-Gathering problem for $n \geq 3$ is simply called Gathering. Gathering is said to be distinct if all robots are initially placed in different positions. An algorithm solving Gathering is said to be \emph{self-stabilizing} if the robots initially have their lights set to arbitrary colors and start their execution from the $\Look$ operation.

Given two points $p,q \in \mathbb{R}^{2}$, we indicate the line segment by $\overline{pq}$ and its length by $|\overline{pq}|$. Let $\mathcal{SS}$ be a configuration or a snapshot. Given $\mathcal{SS}$, $SEC(\mathcal{SS})$ denotes the smallest enclosing circle containing $\mathcal{P(SS)}$, and the length of its diameter and center are denoted by $Diam(\mathcal{SS})$ and $CTR(\mathcal{SS})$, respectively. A longest distance segment ($LDS$, for short) in $\mathcal{SS}$ is a line segment $\overline{pq}$ such that $p,q \in \mathcal{P(SS)}$ and $|\overline{pq}| = max_{x,y\in \mathcal{P(SS)}}|\overline{xy}|$ and the set of the longest distance segments in $\mathcal{SS}$ is denoted by $LDS(\mathcal{SS})$. If $|LDS(SS)|=1$, $LDS$ in $\mathcal{SS}$ is denoted by $\overline{pq}_{\mathcal{SS}}$ and $O(LDS(\mathcal{SS}))$ denotes the set of points that are not within $\overline{pq}_{\mathcal{SS}}$. If $|O(LDS(\mathcal{SS}))|=0$, $\mathcal{SS}$ is called $OnLDS$.

Formally we define {\em color configurations}  as follows. Let $C(t)$ be a
configuration at time $t$, and let $p$ and $q$ be the endpoints of $\overline{pq}_{C(t)}$. The configuration $C(t)$ has a color configuration.
\begin{enumerate}
    \item $\alpha\beta$, if all robots at $p$ have color $\alpha$, all robots at $q$ have color $\beta$ ($\alpha,\beta \in\{A,B\}$) and there are no robots inside $\overline{pq}_{C(t)}$.
    \item $\alpha\gamma\beta$, if all robots at $p$ have color $\alpha$, all robots at $q$ have color $\beta$, all robots at the mid-point of the $\overline{pq}_{C(t)}$ have color $\gamma$ ($\alpha,\beta\,\gamma in\{A,B\}$) and no robots are in other locations.
    \item $\alpha{\mycom{\gamma}{\zeta}}\beta$, if all robots at $p$ have color $\alpha$, all robots at $q$ have color $\beta$, all robots at the mid-point of the $\overline{pq}_{C(t)}$ have color $\gamma$ or $\zeta$ ($\alpha,\beta,\gamma,\zeta \in\{A,B\}$) and no robots are in other locations.
\end{enumerate}

\section{Minimum Number of Lights for Gathering}\label{sec:Impossibility}

The following Theorem~\ref{thm:two-light-insufficient} provides the 
complementing lower bound to the known upper bounds for 
both the case of $\FC$ and $\FS$ robots by showing that two colors of 
lights do not suffice. 
Note that this contrasts with the full-light model, where two 
lights are sufficient for gathering in 
\SSY~\cite{TWK}.
%\SSY.~\cite[Algorithm~2]{TWK}
Moreover, note that the restriction to $\FS$ in the last part of the theorem is inevitable because Rendezvous is possible with 3 colors even in \ASY\ with non-rigid movement~\cite[Thm.~4]{V}.

\begin{theorem}[3 colors of $\FC$ and $\FS$ robots are 
necessary]\label{thm:two-light-insufficient}
Consider the $\FC$ or $\FS$ model working in \SSY\ with rigid 
movement, consistent chirality, and a shared unit. 
%and the minimum distance $\delta$ agreement.
With only two colors of lights, Gathering is impossible for any $n\ge2$. 
Moreover, for the $\FS$ model, it is impossible even with an unlimited number of 
colors if the algorithm is assumed to be self-stabilizing. 
\end{theorem}

\begin{proof}
The case of $\FS$ robots and the case of $\FC$ robots can be proved similarly.
For the sake of simplicity, we assume $n=2$. The proof is easily 
generalized to the case $n>2$ by considering 
a 2-point configuration in which the robots of each point are 
always activated together. 

We show that no algorithm can achieve a gathering for every 
\SSY\
schedule. 
Fix an algorithm for the two robots. We assume that they have the 
same snapshot except that one is rotated by 180 degrees with respect to 
the other. We construct a schedule for 
which the robots with this algorithm cannot gather, round by round. 
For any round that starts with a configuration where the robots do 
not gather if they are both activated, we activate them both. 
If activating them both leads to a gathering, we distinguish two 
cases. If both robots move during their cycle, then we change the 
schedule such that only one robot is 
activated to prevent the gathering. 
The robot to be activated can be chosen arbitrarily; we can 
therefore continue to use this strategy without violating fairness 
to prevent a gathering as long as this case occurs.

The remaining case is a round that achieves gathering with only 
one of the two robots moving, even if both are activated. 
In this case, the two robots behave differently---one is standing still while the other is moving towards it---and must thus
see different internal lights. 
We first keep activating the non-moving robot, which might 
change its color, until it either decides to move or has cycled 
back to a previous color without any movement. (There is no third 
possibility since the number of colors is assumed to be finite. 
Indeed, gathering is always possible with unlimited internal 
memory.~\cite{C04}.) In the 
former case of movement, we continue the strategy as described 
above to prolong the schedule that prevents the gathering. 
The latter case remains, where one of the robots cycles through 
a set of colors without ever moving.

Assume for this paragraph the case of $\FS$ and a self-stabilizing algorithm, that is, that we can impose an initial configuration that is 
identical except for both robots seeing the same color as the internal light of the 
non-moving robot, then we can activate one of the robots arbitrarily often without any movement, until it loops back to a previously used color. We can do the same with other roobt.  
Thus both robots are in a non-moving loop and they never meet. 
This shows that gathering is impossible for $\FS$ robots if we require the 
algorithm to work for any initial configuration where all lights 
are 
set to the same color, possibly different from $A$. 
Thus Gathering is impossible for $\FS$ robots even with an unlimited number of 
colors if the algorithm is assumed to be self-stabilizing. 

We now drop the last paragraph's restriction to $\FS$ and the requirement of self-stabilization, allowing instead only initial configurations 
in which the robots are all set to light $A$. 
In exchange for weakening the adversarial scheduler, we also weaken 
the robots by granting them only the two colors $A$ and $B$. 
In the final round as constructed before, the moving robot 
sees one of the two colors and the non-moving robot the other color. 
In this case they cannot achieve a gathering from the 
initial configuration that is identical except for both robots 
having and thus seeing light $A$: If we keep activating both 
of them forever, then there are four subcases. 
The first one is that they both see the color that lets them move and that they keep the color while doing so. In this case, they will swap their positions forever.  
The second subcase is that they both see the color that lets them stay where they are and that they keep the color when doing this. In this case, they will stay where they are forever. 
The third subcase is that they both see the color that lets them stay, but they change the color while doing so. After one such swap we are in the second or the fourth subcase. 
The fourth and last subcase is that they both see the color that lets them stay where they are, but they change the color when doing so. 
This leads us back to the first or third subcase. 
The robots will thus keep synchronously swapping positions or staying where they are  forever, preventing a gathering. \qed
\end{proof}

Using this theorem, it is shown that all the conditions of previous Gathering algorithms for $\FC$ and $\FS$ robots shown in Table~\ref{tab:Table-Gathering} are necessary.

\begin{theorem}\label{th:app-th1}
\begin{enumerate}
\item[(1)] For the 3-color algorithm~\cite[Algorithm~3]{TWK} for $\FC$ robots with rigid movement in \SSY, the number of used colors is optimal. 
\item[(2)] For the 2-color algorithm~\cite[Algorithm~4]{TWK} for $\FC$ robots with rigid movement in \SSY, the condition of local-awareness cannot be removed.
\item[(3)] For the 2-color algorithm~\cite[Algorithm~5]{TWK} for $\FS$ robots with non-rigid movement in \SSY, the condition of $2\delta$-distant cannot be removed completely.
\end{enumerate}    
\end{theorem}

In the next section, we weaken the condition of $2\delta$-distant in Theorem~\ref{th:app-th1}~(3) instead of removing it completely.

\section{2-Color Gathering Algorithm for $\FS$ Robots}\label{sec:algorithm}
In this section, we show a Gathering algorithm for $\FS$ robots with 2 colors in Non-rigid$(+\delta=)$ with agreements of chirality if we exclude 2 patterns stated below from the initial configurations.

The views of the robots in $\FS$ are the same as those of the robots in $\OB$, so the robots must determine their behavior using these views without colors and their own colors of lights. Thus Gathering algorithms in $\FS$ cannot seem to be constructed without additional knowledge such as distance information. In fact, known Rendezvous algorithms use the minimum distance of moving $\delta$ and/or the unit distance \cite{flocchini2016rendezvous}.

In our Gathering algorithm for $\FS$ robots, we assume the agreement of $d(<\frac{\delta}{4})$ and $\epsilon(<d)$. The prohibited initial configurations are as follows: (1) there are two points $a$ and $b$ such that $d-\frac{\epsilon}{2} \leq |\overline{ab}| < d$, or (2) there are three points $a$, $b$ and $c$ such that $2d-\epsilon \leq |\overline{ac}| < 2d$, and $|\overline{ab}|=|\overline{bc}|$ (Fig.~\ref{fig:NGI}).  These patterns appear in the final phase of the Gathering algorithm and are used to achieve the Gathering using color.
However, when these patterns appear as the initial configuration, the $\FS$ robots cannot distinguish them from the patterns of this final phase.
If initial configurations do not include the two prohibited patterns, we can construct a Gathering algorithm for $\FS$ robots in Non-rigid$(+\delta=)$ and \SSY\ with 2 colors of light.

\begin{figure}[H]
        \centering
        \includegraphics[scale=0.7]{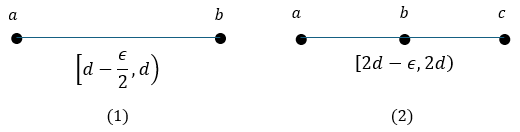}
        \caption{The prohibited initial configurations}
        \label{fig:NGI}
\end{figure}

Our Gathering algorithm for $\FS$ robots (Algorithm \ref{alg1}) consists of three parts.
\begin{enumerate}
    \item From any configuration, we make an $OnLDS$ configuration $C$ such that $|\overline{pq}_C| \geq 2d$ (Algorithm \ref{alg2}).
    \item From any $OnLDS$ configuration $C$ such that $|\overline{pq}_C| \geq 2d$, we make a 2-point configuration $C^{\prime}$ such that %$C^{\prime}=\{a,b\}$ and 
    $2d-\epsilon < |\overline{pq}_{C^{\prime}}| < 2d$ (Algorithm \ref{alg3}).
    \item From any 2-point configuration $C$ such that $2d-\epsilon \leq |\overline{pq}_C| < 2d$, we make a Gathering configuration (Algorithm \ref{alg4}).
\end{enumerate}

\begin{algorithm}[H]
    \caption{Gathering-$\FS$-Robots($r_i$)}
    \label{alg1}
    \begin{algorithmic}[1]
    \REQUIRE Any configuration except the two prohibited patterns, all robots have color $A$.
    \ENSURE Gathering configuration.
    \vskip\baselineskip
    \IF{$\lnot OnLDS$ \OR ($OnLDS$ \AND $2d > |\overline{pq}_{\mathcal{SS}_i}|$)} 
    \STATE ElectLDS$(r_i)$
    \ELSIF{$OnLDS$ \AND $2d \leq |\overline{pq}_{\mathcal{SS}_i}|$} 
    \STATE Adjustment-LDS$(r_i)$
    \ELSIF{$|\mathcal{P}(\mathcal{SS}_i)|=2$ \AND $2d-\epsilon \leq |\overline{pq}_{\mathcal{SS}_i}| < 2d$}
    \STATE Gather$(r_i)$ 
    \ENDIF
    \end{algorithmic}
\end{algorithm}

Note that we do not use colors to solve Cases 1 and 2 and we only use two colors to solve Case 3. The output of $i$ is the input of $i+1$ for $i\in\{1,2\}$ and the explanation of the algorithms is listed in order from 1 to 3. 
The configurations transitions of Algorithm~\ref{alg1} are shown in Fig.~\ref{fig:OLA1}, where nodes denote configurations and a directed edge denotes the transition from a configuration to a configuration. 

\begin{figure}[H]
        \centering
        \includegraphics[scale=0.7]{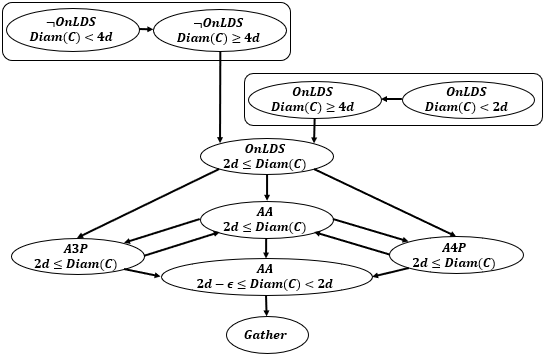}
        \caption{The transition of configurations in Gathering-$\FS$-Robots($r_i$).}
        \label{fig:OLA1}
\end{figure} 

The outline of the behavior of Algorithm \ref{alg1} is explained as follows. From any initial configuration $C_0$ of $\lnot OnLDS$ or $OnLDS$, robots move radially outward by $4d<\delta$ from $CTR(C_0)$ if $Diam(C_0) <4d$, making the configuration $Diam(C_1)\geq4d$, and then robots make the $OnLDS$ configuration $C_2$ such that $Diam(C_2)\geq2d$ (lines 1--2) by using ElectLDS-Preserving-Distance$(r_i)$~\cite{TWK},
making $OnLDS$ preserving the diameter.
When the configuration $C_2$ is obtained, %$OnLDS$ such that $2d\leq Diam(C_2)$, 
robots located not at endpoints move to endpoints and robots located at endpoints stay and then a 2-point configuration $C_3$ such that $2d\leq Diam(C_3)$ results through special patterns $A3P$ or $A4P$ (lines 3--4). 
Then the robots reduce the diameter and make
%When the configuration $C_3$ is a 2-point configuration or $A3P_{far}$ or $A4P_{far}$ such that $2d\leq Diam(C_3)$, robots located at endpoints move radially inwards by $\frac{\epsilon}{2}(<d)$ it becomes 
%a 2-point configuration $C_4$ such that $Diam(C_4)=Diam(C_3)-\epsilon(-\frac{\epsilon}{2})$ (lines 3-4). Thus it becomes 
a 2-point configuration $C_4$ such that $2d-\epsilon \leq Diam(C_4) < 2d$ through special patterns $A3P$ or $A4P$. %When the configuration $C_5$ is a 2-point configuration or $A3P_{near}$ or $A4P_{near}$ such that $d\leq Diam(C_5)<2d-\epsilon$, robots located not at endpoints move radially outwards by $\frac{\epsilon}{2}(<d)$ it becomes a 2-point configuration $C_6$ such that $Diam(C_6)=Diam(C_5)+\epsilon(+\frac{\epsilon}{2})$ (lines 3-4). 
%Thus it becomes a 2-point configuration $C_7$ such that $2d-\epsilon \leq Diam(C_7) < 2d$. 
When a 2-point configuration $C_4$ is obtained, %such that $2d-\epsilon \leq Diam(C_7) < 2d$, 
robots with color $A$ change its color to $B$ and move to the midpoint and robots with $B$ color stay,  Gathering is achieved (lines 5--6).

Algorithm \ref{alg2} is based on ElectLDS-Preserving-Distance$(r_i)$~\cite[Algorithm~8]{TWK} and produces the unique $LDS$ with its length at least 2d unless it produces a Gathering configuration. 

\begin{lemma}\cite{TWK}\label{lem:EPD}
Using ElectLDS-Preserving-Distance, if $C(t)$ is any configuration, % with $Diam(C(t))$, 
then there exists a time $t^{\prime}>t$ such that $C(t)\rightarrow^{\ast}C(t^{\prime})$, $C(t')$ is a Gathering configuration or an $OnLDS$  configuration with $Diam(C(t^{\prime})) \geq Diam(C(t))/2$.
\end{lemma}

In the following algorithms, the expression $[p,q] + \alpha - \beta$ denotes the point at distance $\alpha - \beta$ from point $p$ on the line segment $\overline{pq}$.
%In Algorithm \ref{alg2}, $p_{out}$ letting any activated robot move radially outwards by $4d < \delta$ from the $CTR(\mathcal{SS})$.

\begin{algorithm}[H]
    \caption{ElectLDS$(r_i)$}
    \label{alg2}
    \begin{algorithmic}[1]
    \REQUIRE $\lnot OnLDS$ or ($OnLDS$ and $2d > |\overline{pq}_{\mathcal{SS}_i}|$)
    \ENSURE $OnLDS$ and $2d\leq |\overline{pq}_{\mathcal{SS}_i}|$
    \vskip\baselineskip
    \IF{$Diam(\mathcal{SS}_i) < 4d$}
    \STATE $des_{i} \leftarrow [p_i,CTR(\mathcal{SS}_i)] + 4d$ \quad //move $4d$ outward from $CTR(\mathcal{SS}_i)$
    \ELSIF{$Diam(\mathcal{SS}_i) \geq 4d$}
    \STATE ElectLDS-Preserving-Distance$(r_i)$ \quad //\cite[Algorithm~8]{TWK}
    \ENDIF
    \end{algorithmic}
\end{algorithm}

Lemma \ref{lem:1} ensures that an $OnLDS$ configuration $C(t)$ with $|\overline{pq}_{C(t)}|\geq 2d$ is reached.

\begin{lemma}\label{lem:1}
If $|O(LDS(C(t)))| \geq 0$, then there exists a time $t^{\prime}>t$ such that $C(t)\rightarrow^{\ast}C(t^{\prime})$, $C(t')$ is a Gathering configuration or an $OnLDS$ configuration with $|\overline{p'q'}_{C(t^{\prime})}| \geq 2d$.
\end{lemma}

\begin{proof}
Whenever the robots are not configured $OnLDS$, they consider $SEC(C(t))$ and $LDS(C(t))$. If $Diam(C(t))$ is below $4d$, then it is increased by letting any activated robot move radially outwards by $4d < \delta$ from $CTR(C(t))$. This increases $Diam(C(t))$ by at least $4d$ in one step for the following reason. For any point on or in the smallest enclosing circle $SEC(C(t))$, there is another point at an angle of at least $90$ degrees when viewed from the $CTR(C(t))$ because all points would lie on one side of some diameter of the $SEC(C(t))$, contradicting its minimality. The outward movements thus cannot cancel each other out; they will add up to an increase of at least $4d$. We now consider the second case where $Diam(C(t))$ is at least $4d$. Then the mentioned algorithm is used to reach a Gathering or a $OnLDS$ configuration $C(t'))$ with $Diam(C(t')) \geq 2d$ according to Lemma~\ref{lem:EPD}.\qed
\end{proof}

Next, we make a 2-point configuration $C$ satisfying $2d-\epsilon \leq |\overline{pq}_{C}| < 2d$ from any $OnLDS$ configuration. This adjustment task is performed using Algorithm~\ref{alg3}. In the algorithm for robot $r_i$, the endpoints of $LDS$ are denoted by $p_{n}$ and $p_f$ for position $p_i$ of $r_i$, where $p_n$ is the nearest endpoint from $p_i$ and $p_f$ is the farthest one from $p_i$. Note that if $p_i$ is one of the endpoints, then $p_i=p_n$ and $p_f$ is the other endpoint. %In Algorithm \ref{alg3}, $[p,q]+\alpha-\beta$ letting the point $P$ with distance $\alpha-\beta$ on the line segment $\overline{pq}$.

In Algorithm~\ref{alg3}, $A3P$ and $A4P$ denote the following predicates (Fig.~\ref{fig:A3PA4P}). When configuration $C(t)$ is 2-point configuration, if robots move on one endpoint only, $C(t+1)$ satisfies $A3P$ and if robots move on both endpoints, $C(t+1)$ satisfies $A4P$.

\begin{description}
    \item[$A3P$ :] $|\mathcal{P}(\mathcal{SS}_i)|=3, \mathcal{P}(\mathcal{SS}_i)=\{p_n,p_{m_1},p_f\}, |\overline{p_np_{m_1}}|\neq|\overline{p_{m_1}p_f}|, |\overline{p_np_{m_1}}|=\frac{\epsilon}{2}, 2d-\epsilon \leq |\overline{p_{m_1}p_f}|$, %(have a mirror image) ??
    %\item[$A3P_{near}$ :] $|\mathcal{P}(\mathcal{SS}_i)|=3, \mathcal{P}(\mathcal{SS}_i)=\{p_n,p_{m_1},p_f\}, |\overline{p_np_{m_1}}|\neq|\overline{p_{m_1}p_f}|, |\overline{p_np_{m_1}}|=\frac{\epsilon}{2}, d-\frac{\epsilon}{2} \leq |\overline{p_{m_1}p_f}|<2d-\epsilon$, (have a mirror image)
    \item[$A4P$ :] $|\mathcal{P}(\mathcal{SS}_i)|=4, \mathcal{P}(\mathcal{SS}_i)=\{p_n,p_{m_1},p_{m_2},p_f\}, |\overline{p_np_{m_1}}|=|\overline{p_{m_2}p_f}|=\frac{\epsilon}{2}, 2d-2\epsilon \leq |\overline{p_{m_1}p_{m_2}}|$
    %\item[$A4P_{near}$ :] $|\mathcal{P}(\mathcal{SS}_i)|=4, \mathcal{P}(\mathcal{SS}_i)=\{p_n,p_{m_1},p_{m_2},p_f\}, |\overline{p_np_{m_1}}|=|\overline{p_{m_2}p_f}|=\frac{\epsilon}{2}, d-\epsilon \leq |\overline{p_{m_1}p_{m_2}}|<2d-\epsilon$
\end{description}

\begin{figure}[H]
        \centering
        \includegraphics[scale=0.7]{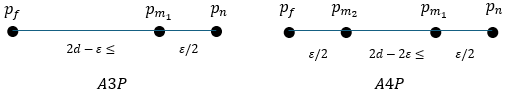}
        \caption{Configurations in $A3P$ and $A4P$.}
        \label{fig:A3PA4P}
\end{figure}

%In Algorithm \ref{alg3}, expression $[p,q] + \alpha - \beta$ means the point from point $p$ with distance $\alpha - \beta$ on the line segment $\overline{pq}$.

\begin{algorithm}[H]
    \caption{Adjustment-LDS$(r_i)$}
    \label{alg3}
    \begin{algorithmic}[1]
    \REQUIRE $OnLDS$ and $2d \leq |\overline{pq}_{\mathcal{SS}_i}|$
    \ENSURE $|\mathcal{P}(\mathcal{SS}_i)|=2$ and $2d-\epsilon \leq |\overline{pq}_{\mathcal{SS}_i}| < 2d$.
    \vskip\baselineskip
    \IF{$(((|\mathcal{P}(\mathcal{SS}_i)|=2)$ \AND $(2d \leq |\overline{pq}_{\mathcal{SS}_i}|))$ \OR $A3P$ \OR $A4P)$ \AND $(p_{i}=p_{n})$}
    \STATE $des_{i} \leftarrow [p_n,p_f]-\frac{\epsilon}{2}$ \quad //$\frac{\epsilon}{2}$ inward
    %\ELSIF{$((|\mathcal{P}(\mathcal{SS}_i)|=2)$ \AND $(d\leq|\overline{pq}_{\mathcal{SS}_i}|<2d-\epsilon)$ \AND $(p_{i}=p_{n}))$ \OR $((A3P_{near}$ \OR $A4P_{near})$ \AND $(p_{i}=p_{m_{1}}))$}
    %\STATE $des_{i} \leftarrow [p_n,p_f]+\frac{\epsilon}{2}$ \quad //$\frac{\epsilon}{2}$ outward
    \ELSIF{$((|\mathcal{P}(\mathcal{SS}_i)|\geq3)$}
    \STATE $des_{i} \leftarrow p_{n}$
    \ENDIF
    \end{algorithmic}
\end{algorithm}

\begin{lemma}\label{lem:3}
    If $C(t)$ is an $OnLDS$ configuration with $|\overline{pq}_{(C(t))}| \geq 2d$, then by Algorithm~\ref{alg3}, there are times $t'$ and $t''$ ($t''>t'>t$) such that $C(t)\rightarrow^{\ast}C(t^{\prime})$, $C(t')$ satisfies one of the following, and $C(t')\rightarrow^{\ast}C(t^{\prime\prime})$, $C(t'')$ is a 2-point configuration with $2d-\epsilon \leq |\overline{p''q''}_{C(t'')}| < 2d$.
    \begin{enumerate}
        \item $|\mathcal{P}(C(t^{\prime})|=2$ and $|\overline{p^{\prime}q^{\prime}}_{C(t^{\prime})}|=|\overline{pq}_{C(t)}|$
        \item $|\mathcal{P}(C(t^{\prime})|=2$ and $|\overline{p^{\prime}q^{\prime}}_{C(t^{\prime})}|=|\overline{pq}_{C(t)}|-\frac{\epsilon}{2}$
        \item $|\mathcal{P}(C(t^{\prime})|=2$ and $|\overline{p^{\prime}q^{\prime}}_{C(t^{\prime})}|=|\overline{pq}_{C(t)}|-\epsilon$
    \end{enumerate}
\end{lemma}

\begin{proof}
    When $|\mathcal{P}(C(t))| \geq 3$, an activated robot $r_i$ not positioned at the endpoints($p_i \neq p_n$), moves toward $p_n$ (line 6).
    \begin{enumerate}
        \item When all robots reach $p_n$, $|\overline{p^{\prime}q^{\prime}}_{C(t^{\prime})}|$ remains unchanged and becomes a 2-point configuration.
        \item Before all robots reach $p_n$, $C(t)$ may become $A3P$. When $C(t)$ is $A3P$, the robot at $p_n$ closest to $p_{m_1}$ moves toward $p_{m_1}$ (line 2), resulting in a 2-point configuration with $|\overline{p^{\prime}q^{\prime}}_{C(t^{\prime})}|$ reduced by $\frac{\epsilon}{2}$.
        \item Before all robots reach $p_n$, $C(t)$ may become $A4P$. When $C(t)$ is $A4P$, the robot at $p_n$ moves toward $p_{m_1}$ (line 2), resulting in a 2-point configuration with $|\overline{p^{\prime}q^{\prime}}_{C(t^{\prime})}|$ reduced by $\epsilon$.
    \end{enumerate}
    After $t'+1$, the activated robot moves $\frac{\epsilon}{2}$ inward (line 2). By iterating the above steps, a 2-point configuration $C(t^{\prime\prime})$ with $2d-\epsilon \leq |\overline{p''q''}_{C(t'')}| < 2d$ is reached.\qed
\end{proof}

%Special cases occur when the initial configuration is $A3P_{near}$ or $A4P_{near}$.

%\begin{lemma}\label{lem:4}
    %If the initial configuration $C(t)$ is $A3P_{near}$ or $A4P_{near}$, then there is a time $t'>t$ such that $C(t^{\prime})$ is a 2-point configuration with $2d-\epsilon \leq |\overline{pq}_{C(t')}| < 2d$.
%\end{lemma}

%\begin{proof}
    %Only the activated inner robot can move (line 4). Since the inner robot moves toward the position $p_n$, $C(t_1)$ is a 2-point configuration with $d \leq|\overline{pq}_{C(t_1)}|<2d-\epsilon$. After  $t_1+1$, the activated robot moves $\frac{\epsilon}{2}$ outward (line 4). And the same steps are repeated thereafter. Thus, there is a time $t'$ such that $C(t^{\prime})$ is a 2-point configuration with $2d-\epsilon \leq |\overline{pq}_{C(t')}|<2d$.\qed
%\end{proof}

By Lemmas~\ref{lem:1}--\ref{lem:3}, it is guaranteed that there is a time $t$ such that a 2-point configuration $C(t)$ with $2d-\epsilon \leq |\overline{pq}_{C(t)}|<2d$ is obtained. 

Lastly, Algorithm~\ref{alg4} solves Gathering in $\FS$ if the initial configurations satisfy $|\mathcal{P(SS}_i)| = 2$ and $2d-\epsilon \leq |\overline{pq}_{\mathcal{SS}_i}|<2d$. Since $d \leq \frac{\delta}{4}$, every movement in Algorithm \ref{alg4} is the same as the rigid one. The following lemma is easily verified for Algorithm \ref{alg4}.

\begin{algorithm}[H]
    \caption{Gather$(r_i)$}
    \label{alg4}
    \begin{algorithmic}[1]
    \REQUIRE $|\mathcal{P}(\mathcal{SS}_i)|=2$ and $2d-\epsilon \leq |\overline{pq}_{\mathcal{SS}_i}| < 2d$, all robots have color $A$.
    \ENSURE Gathering configuration.
    \vskip\baselineskip
    \IF{$(2d-\epsilon \leq |\overline{pq}_{\mathcal{SS}_i}| < 2d)$ \AND $((|\mathcal{P}(\mathcal{SS}_i)|=2)$ \OR $((|\mathcal{P}(\mathcal{SS}_i)|=3)$ \AND $(|\overline{p_{n}p_{m}}|=|\overline{p_{m}p_{f}}|)))$ \AND $(l_{i}=A)$}
    \STATE $l_{i} \leftarrow B$
    \STATE $des_{i} \leftarrow p_{m}$
    \ELSIF{$(d-\frac{\epsilon}{2} \leq |\overline{pq}_{\mathcal{SS}_i}| < d)$ \AND $(|\mathcal{P}(\mathcal{SS}_i)|=2)$ \AND $(l_{i}=A)$}
    \STATE $l_{i} \leftarrow B$
    \STATE $des_{i} \leftarrow p_{f}$
    \ELSIF{$l_i=B$} 
    \STATE $l_i \leftarrow B$
    \STATE $des_{i} \leftarrow p_{i}$ \quad //stay
    \ENDIF
    \end{algorithmic}
\end{algorithm}

\begin{lemma}\label{lem:5}
    If $C(t)$ is a 2-point configuration with $|\overline{pq}_{C(t)}|$ satisfying $2d-\epsilon \leq |\overline{pq}_{C(t)}| < 2d$, %then there exists times $t$, $t^{\prime}$,$t^{\prime\prime}$ and $t^{\prime\prime\prime}$ ($t^{\prime\prime\prime}>t^{\prime\prime}>t^{\prime}>t$) such that $C(t^{\prime\prime\prime})$ is a Gathering configuration.
    there is a time $t^{\prime}>t$ such that $C(t)\rightarrow^{\ast}C(t^{\prime})$ and $C(t')$ is a Gathering configuration.
\end{lemma}

\begin{proof}
    From such a 2-point configuration, the activated robots move to the midpoint while changing light from $A$ to $B$, and stay there afterwards. So when all the robots at the endpoints are activated at the same time, Gathering is achieved. When some of the robots at the endpoints are activated, $C(t_1)$ reaches a 3-point configuration which has the the color configuration $ABA$. In $C(t_1)$, the robots with color $A$ move the same as before, so when the number of robots at the endpoints decreases and there are no more robots with color $A$ at one endpoint, we have a 2-point configuration $C(t_2)$ with the color configuration $AB(BA)$. Then, since the 2-point configuration satisfies $d-\frac{\epsilon}{2} \leq |\overline{pq}_{C(t_2)}| < d$, the robots with color $A$ move to the other endpoint while changing their light from $A$ to $B$, and there is a time $t'$ such that $C(t^{\prime})$ is a Gathering configuration.\qed
\end{proof}

The prohibited patterns appear in this case. If the initial configuration is a 2-point configuration $C$ such that $C=\{a,b\}$ and $d-\frac{\epsilon}{2} \leq |\overline{ab}| < d$ with the color configuration $AA$, the robots with color $A$ move to the other endpoint while changing their lights from $A$ to $B$, and stay there thereafter, and $C_1$ reaches a 2-point configuration with the color configuration $BB$, and thereafter the algorithm terminates and Gathering is not achieved. If the initial configuration is a 3-point configuration $C$ such that $C=\{a,b,c\}$ and $2d-\epsilon \leq |\overline{ac}| < 2d$ and $|\overline{ab}|=|\overline{bc}|$ with the color configuration $AAA$, the activated robots located at the endpoints move to the midpoint while changing their lights from $A$ to $B$, and stay there thereafter, the activated robots not located at the endpoints move to the endpoints without changing their lights, $C_2$ reaches a 3-point configuration with the color configuration $A{\mycom{A}{B}}A$ or a 2-point configuration with the color configuration $A{\mycom{A}{B}}({\mycom{A}{B}}A)$. When $C_2$ is the 3-point configuration with the color configuration $A{\mycom{A}{B}}A$, the robots with color $A$ move the same as before, so when the number of robots at the endpoints decreases and there are no more robots with color $A$ at one endpoint, $C^{\prime}_2$ reaches the 2-point configuration with the color configuration $A{\mycom{A}{B}}({\mycom{A}{B}}A)$. When $C^{\prime}_2$ is a 2-point configuration with the color configuration $A{\mycom{A}{B}}({\mycom{A}{B}}A)$, for the same reason as in a 2-point configuration with $d-\frac{\epsilon}{2} \leq |\overline{ab}| < d$ with the color configuration $AA$, $C^{\prime\prime}_2$ reaches a 2-point configuration with the color-configuration $BB$, and thereafter the algorithm terminates and Gathering is not achieved.

We obtain the following theorem by Lemma \ref{lem:1}-\ref{lem:5}.

\begin{theorem}\label{thm:two-FSTA-gather}
Let $d < \frac{\delta}{4}$ and $\epsilon < d$ be agreed upon by the robots. Gathering is solvable in $\FS$, Non-Rigid($+\delta=$), and \SSY if robots have 2 colors, set view and agreement on chirality, and the following initial configuration excepted.
\begin{enumerate}
    \item 2-point configuration $C=\{a,b\}$ such that $d-\frac{\epsilon}{2} \leq |\overline{ab}| < d$
    \item 3-point configuration $C=\{a,b,c\}$ such that $2d-\epsilon \leq |\overline{ac}| < 2d$ and $|\overline{ab}|=|\overline{bc}|$
\end{enumerate}
\end{theorem}

\section{Concluding Remarks}\label{sec:conclusion}

In this paper, we have shown that Gathering is impossible for $\FC$ and $\FS$ robots with 2-color lights in \SSY, even assuming agreement on chirality and the minimum distance $\delta$, and consequently, that the conditions imposed on previously developed Gathering algorithms for $\FC$ and $\FS$ robots are necessary. We have also improved the condition imposed on the Gathering algorithm for $\FS$ robots with 2-color lights.

Interesting open questions are developing unconditional Gathering algorithms for $\FS$ robots with more than two colors in \SSY, and/or Gathering algorithms for $\FC$ or $\FS$ robots in \ASY.

\bibliographystyle{plainurl}
\bibliography{referenceorg}

\end{document}